\documentclass[12pt]{article}
\usepackage{graphics,amsmath,amssymb}
\usepackage{amsthm}
\usepackage{amsfonts}
\usepackage{latexsym}
\usepackage{amscd}

\newtheorem{theorem}{Theorem}[section]
\newtheorem{corollary}{Corollary}[section]

\begin{document}
\begin{center}
\title{How to find the least upper bound on the van der Waerden Number $W(r, k)$ that is some integer Power of the coloring Integer $r$}
\author{\textbf{Robert J. Betts}}
\maketitle
\emph{The Open University\\Postgraduate Department of Mathematics and Statistics~\footnote{During 2012--2013, when the Author was working on an earlier draft.}\\ (Main Campus) Walton Hall, Milton Keynes, MK7 6AA, UK\\
Robert\_Betts@alum.umb.edu}
\end{center}
\begin{abstract}
What is a least integer upper bound on van der Waerden number $W(r, k)$ among the powers of the integer $r$? We show how this can be found by expanding the integer $W(r, k)$ into powers of $r$. Doing this enables us to find both a least upper bound and a greatest lower bound on $W(r, k)$ that are some powers of $r$ and where the greatest lower bound is equal to or smaller than $W(r, k)$. A finite series expansion of each $W(r, k)$ into integer powers of $r$ then helps us to find also a greatest real lower bound on any $k$ for which a conjecture posed by R. Graham is true, following immediately as a particular case of the overall result.,,
\end{abstract}
\section{Introduction}
We shall understand $W(r, k)$ to be that van der Waerden number that is the least integer for which 
\begin{equation}
[1, W(r, k)] \subset \mathbb{R}^{+},
\end{equation}
contains an arithmetic progression (AP) of integers with $k$ terms, where these integers have some monochromatic coloring from among $r$ different colors. Another way to see this is to apply van der Waerden's Theorem to this context in the following manner: Suppose we partition the integer interval 
$$
[1, W(r, k)] = \bigcup_{j = 1}^{r}W_{j},
$$
into $r$ disjoint subsets \(W_{j} \: \: \forall j \in [1, r]\), where each of these integer subsets $W_{j}$ has its own integer elements all painted with the same color (called a ``monochromatic" coloring) and by exactly one color chosen from the $r$ different colors. Then no less than one of these disjoint integer subsets $W_{j}$ will have an AP of $k$ terms among some of its integer elements that are all painted with the same color for that subset.\\  
\indent Elsewhere in the literature~\cite{Rabung and Lotts}, one sees $k$ and $l$, where $l$ denotes the length of the AP and $k$ denotes the number of colors, but this latter notation is not universal. So the reader should be aware we instead use the expression $W(r, k)$ instead of $W(k, l)$. So here we choose to use $r$ for the number of integer colorings and $k$ for the length of the AP.\\
\indent Also one sees elsewhere in the literature that van der Waerden numbers $W(r, k)$ are treated as primitive recursive functions~\cite{Shelah},~\cite{Graham and Spencer}. Moreover the upper bounds on them previously have been derived with Ackermann() and Tower() functions~\cite{Graham and Spencer}. Although some researchers have attained good results with the use of primitive recursive functions, an approach which seems to have been inspired by B. van der Waerden's much earlier approaches to describe $W(r, k)$~\cite{van der Waerden},~\cite{Graham and Spencer},~\cite{Graham and Rothschild},~\cite{Shelah}, the point should not be lost upon the reader that here in this paper \emph{we do not treat van der Waerden numbers in this fashion, meaning as primitive recursive functions} that are bounded above or below by some other recursive functions. Here we are interested in van der Waerden numbers $W(r, k)$ solely as positive integers on the real line or as natural numbers, not as primitive recursive functions or partial functions that are computable in the sense meant either by Alonzo Church or Alan Turing~\cite{Davis} (See pages 49--51), or that might appear in the theory of computable functions treated by M. Davis~\cite{Davis}. In our goal we are reminded of how even Richard Dedekind~\cite{Dedekind},~\cite{Joyce}, drew a clear demarcation line of sorts between the natural numbers themselves \(1, \: 2, \: 3, \ldots, \in \mathbb{N}\), and any discrete functions $\varphi$ that have various rules of assignment and that are defined on $\mathbb{N}$. \\
\indent Our aim in part is to provide future computational number theorists with new tools by which to find, perhaps in the future, intervals on the real line in which there lie van der Waerden numbers the values of which at present are unknown, such as $W(2, 7)$ and $W(2, 10)$ (See Section 6). So with all due respect to those who have been using recursive functions to find upper bounds on $W(r, k)~\cite{Shelah}$, in this paper we shall treat van der Waerden numbers $W(r, k)$ as integers or natural numbers, such as those investigated by elementary number theorists, where each integer van der Waerden number $W(r, k)$ lies bounded above and below somewhere by other integers on $\mathbb{R}$.  
\section{A greatest positive lower Bound on $k$, such that \(W(r, k) < r^{k^{2}}\) is true always}
In this Section the reader should be aware that by the expressions ``the least upper bound on $W(r, k)$ that is an integer power of $r$" and the ``greatest lower bound on $W(r, k)$ that is an integer power of $r$," we mean specifically two integer powers of $r$ of the forms $r^{a}$, $r^{b}$, for which 
$$
W(r, k) \in [r^{a}, r^{b}),
$$
is true and where $r^{b}$ is the least upper bound and $r^{a}$ is the greatest lower bound for all the integers in the set $[r^{a}, r^{b})$, where \(a, b, a < b\) are both some positive integer exponents. \\
\indent For any positive integer $r$ and for any integer exponents \(x > 1\), consider a set of integer powers $r^{x}$, such as 
$$
\{r, r^{2}, r^{3}, \ldots\}.
$$
Let 
$$
\{r^{x}, r^{x + 1}, \ldots\} \subset \{r, r^{2}, r^{3}, \ldots\},
$$
be a proper subset, where \(W(r, k) < r^{x}\). What is the greatest lower bound $r^{x}$ on this proper subset for which the inequality \(W(r, k) < r^{x}\) holds, where $r^{x}$ is the least upper bound on $W(r, k)$ that is some power of the integer $r$ for the positive exponent $x$? By that we mean 
$$
r^{x} = l.u.b \{r, r^{2}, r^{3}, \ldots r^{x}\} \ni W(r, k) < r^{x}.
$$
It follows immediately that if $r^{x}$ is the greatest lower bound on the subset \(\{r^{x}, r^{x + 1}, \ldots\}\) so that \(W(r, k) < r^{x}\) where $r^{x}$ is the least upper bound on the subset \(\{r, r^{2}, r^{3}, \ldots r^{x}\}\) and for which the inequality is true, then \(r^{x - 1} \leq W(r, k) < r^{x}\) for some two integer exponents $x - 1$, and $x$ and there also is some positive real number exponent \(\delta \in [x - 1, x)\) such that \(r^{x - 1} \leq W(r, k) < r^{x}\), where \(W(r, k) = r^{\delta}\). In this paper we characterize the three exponents $x - 1$, $\delta$, $x$ to show how to answer a longstanding conjecture about such upper bounds on $W(r, k)$. \\
\indent A conjecture raised by R. Graham, J. Spencer and B. Rothschild~\cite{Graham1},~\cite{Graham2},\cite{Graham and Spencer}~\cite{Graham and Rothschild}, is that
\begin{equation}
W(2, k) < 2^{k^{2}}.
\end{equation}
We prove the stronger result \(W(r, k) < r^{k^{2}}\) to be true for any \(r > 1\) and for all integer $k$ that has a particular lower bound.\\
\indent Let $r$ be  any positive integer for which $[1, W(r, k)]$ has an AP of $k$ arbitrary terms and where $W(r, k)$ is the least such integer~\cite{van der Waerden}~\cite{Graham1},~\cite{Graham2},~\cite{Graham and Spencer},~\cite{Graham and Rothschild},~\cite{Landman and Robertson},~\cite{Landman and Culver},~\cite{Khinchin}. Any positive integer $N$ greater than $r$ has some finite power series expansion~\cite{Rosen},~\cite{Abramowitz}, 
\begin{equation}
N = b_{n}r^{n} + b_{n - 1}r^{n - 1} + \cdots + b_{0} \in [r^{n}, r^{n + 1}),
\end{equation}
where 
\begin{equation}
b_{n}, b_{n - 1}, \cdots b_{0} \in [0, r - 1], \: b_{n} \in [1, r - 1], 
\end{equation}
and where $n$ is the least positive integer exponent for which Eqtn. (3) holds. Now substitute $W(r, k)$ for $N$ when \(W(r, k) = N\). Then
\begin{equation}
W(r, k) = b_{n}r^{n} + b_{n - 1}r^{n - 1} + \cdots + b_{0} \in [r^{n}, r^{n + 1}).
\end{equation}
For those readers who might require it we provide a formal Definition as follows:\\
\indent \textbf{Definition:} \emph{For each van der Waerden number $W(r, k)$, we define the positive integer exponent \(n \in \mathbb{N}\) as being the smallest positive integer exponent for which $r^{n}$ divides $W(r, k)$, $r^{n + 1}$ does not divide $W(r, k)$ and for which $W(r, k)$ has the finite power series expansion given in Eqtn. (5)}. \\
\indent From this Definition it follows automatically that for each van der Waerden number $W(r, k)$, we have \(W(r, k) \in [r^{n}, r^{n + 1})\).\\
\indent The least upper bound on the set $[r^{n}, r^{n + 1})$ that is an integer power of $r$ is $r^{n + 1}$ and the greatest lower bound of this set that is some integer power of $r$ is $r^{n}$, where \(W(r, k) \in [r^{n}, r^{n + 1})\) (See again Eqtns. (4)--(5)). We can identify now two of the integer powers we discussed previously, since now we recognize \(x - 1 = n\) and \(x = n + 1\).
\subsection{The least upper Bound on any integer Powers within the Interval $[r^{n}, r^{n + 1})$}
By applying substitution again with $W(r, k)$ for $N$ in Eqtn. (3) when \(W(r, k) = N\) we get \(W(r, k) \in [r^{n}, r^{n + 1})\). Let \(\delta(r, k) = \frac{\log W(r, k)}{\log r}\), \(\delta(r, k) \in \mathbb{R}^{+}\) be that positive real number exponent~\cite{Betts}, for which \(W(r, k) = r^{\delta(r, k)}\) within the interval $[r^{n}, r^{n + 1})$. When it comes to integer powers then, the integer power $r^{n + 1}$ actually is the least upper bound on all possible integer powers within $[r^{n}, r^{n + 1})$. We can claim this because for one thing if \(a^{x} \in [r^{n}, r^{n + 1})\) \(\forall x \in [n, n + 1)\) where $a$ and $x$ are both positive integers, this integer power still is smaller than $r^{n + 1}$ although it might be equal to or greater than $r^{n}$. Also if it turns out the real positive exponent $\delta(r, k)$ is an integer then it cannot be equal to any integer in the interval $[n, n + 1)$ other than $n$. Furthermore there are no other integer exponents included within the interval $[n, n + 1)$ other than $n$. So all this indicates that the integer power $r^{n + 1}$ actually is the least upper bound on any integer powers $a^{x}$ in the set $[r^{n}, r^{n + 1})$ even when \(a = r, x = n\) or when \(a = r, x = \delta(r, k)\). \\
\indent Hence if \(W(r, k) < r^{k^{2}}\) is to be true for $k$, we must ask ourselves the following two questions:
\begin{enumerate}
\item What is the size of the integer exponent $k^{2}$ in comparison to the size of the integer exponent $n + 1$?\\
\item What is the size of the integer power $r^{k^{2}}$ in comparison to the size of the integer power $r^{n + 1}$?
\end{enumerate}
If we was to assume that \(k^{2} \in [1, n + 1)\) then we have at once that \(r^{k^{2}} \leq r^{n} \leq W(r, k) < r^{n + 1}\) for all \(k \in [1, \sqrt{n + 1})\), so that we do not get at all the desired inequality \(W(r, k) < r^{k^{2}}\). Therefore the inequality \(W(r, k) < r^{k^{2}}\) is not going to hold for any positive integer $k$ with an integer square of $k^{2}$, unless it is for all \(k \geq \sqrt{n + 1}\), such that then we will get \(r^{n} \leq W(r, k) < r^{n + 1} \leq r^{k^{2}}\), where $r^{n + 1}$ is the least integer power upper bound on the set $[r^{n}, r^{n + 1})$, where \(W(r, k) \in [r^{n}, r^{n + 1})\). \\
\indent Basically \emph{even from Theorem 3.1 alone} (See Theorem 3.1, next Section), the following one derives automatically for any arbitrary positive integer \(k > 2\) such that, for positive integer $a$, $d$ such that \(a, a + d, a + 2d, \ldots, a + (k - 1)d \in [1, W(r, k)]\):
$$
k \in [\sqrt{n + 1}, \infty) \Longrightarrow r^{n} \leq W(r, k) < r^{n + 1} \leq r^{k^{2}}.
$$
\indent Our approach helps us to derive \(2^{\delta(2, 7)} = W(2, 7) \leq 2^{48}\), \(\delta(2, 7) = \frac{\log_{2}W(2, 7)}{\log_{2}2}\), where \(\delta(2, 7) > \log_{2} 3703\)~\cite{Rabung and Lotts}, and \(W(2, 7) = 2^{\delta(2, 7)}\).\\
\indent Ordinarily in base $r$ arithmetic~\cite{Rosen},~\cite{Abramowitz} the van der Waerden number would be expressed as
\begin{equation}
(b_{n}b_{n - 1}\cdots b_{0})_{r}.
\end{equation}
For example for \(r = 2\) and \(r = 3\) respectively Eqtn. (6) would be expressed either in binary or ternary notation. However the sum of powers in Eqtn. (3) and in Eqtn. (5) still can represent a sum of powers that adds up to the integer $W(r, k)$~\cite{Betts}. We restrict our attention to this latter paradigm.\\
\section{Van der Waerden Numbers $W(r, k^{\prime})$, $W(r, k)$, where \(W(r, k^{\prime}) < W(r, k)\), \(k^{\prime} < k\)}
Next we prove Theorem 3.1, that \(W(r, k) < r^{k^{2}}\) is true necessarily for all integer \(k \in [\sqrt{n + 1}, \infty)\).
\begin{theorem}
Let 
\begin{equation}
W(r, k) = b_{n}r^{n} + b_{n - 1}r^{n - 1} + \cdots + b_{0},
\end{equation}
where 
\begin{equation}
r^{n} \leq W(r, k) < r^{n + 1}. 
\end{equation}
Let 
$$
W(r, k^{\prime}) = b_{n^{\prime}}^{\prime}r^{n^{\prime}} + b_{n^{\prime} - 1}^{\prime}r^{n^{\prime} - 1} + \cdots + b_{0}^{\prime},
$$ 
\(n^{\prime} < n\), be another van der Waerden number such that for integer $k^{\prime}$, \(W(r, k^{\prime}) < r^{n} \leq W(r, k)\) is true \(\forall k^{\prime} \in [1, \sqrt{n + 1})\), so that the interval $[1, W(r, k^{\prime})]$ contains an AP of \(k^{\prime} < k\) terms. Then the inequality 
\begin{equation}
W(r, k) < r^{k^{2}},
\end{equation}
is true for any \(k \in [\sqrt{n + 1}, \infty)\).
\end{theorem}
\begin{proof}
One ought to be able to admit at the very least, that the third inequality that appears to the far right in the triple inequality
$$
r^{n} \leq W(r, k) < r^{n + 1} \leq r^{k^{2}},
$$
is true if and only if \(k \geq \sqrt{n + 1}\), a condition which does happen to hold actually if one bothers to check, for all the known van der Waerden numbers \(W(2, 3), W(2, 4), W(2, 5), W(2, 6), W(3, 3), W(3, 4), W(4, 3)\) to date (See Table A). In fact the condition \(k \geq \sqrt{n + 1}\) is both a necessary and sufficient condition for which the inequality \(W(r, k) < r^{k^{2}}\) will be true for any \(k > 2\) such that the interval $[1, W(r, k)]$ on $\mathbb{R}$ will contain an AP of $k$ terms. Nevertheless we shall include a second argument that shows that \(k < \sqrt{n + 1}\) cannot hold always for all $W(r, k)$. \\
\indent First there are infinitely many van der Waerden numbers. This means the relative complement 
\begin{equation}
\mathbb{R}^{+} - [1, W(r, k^{\prime})] = (W(r, k^{\prime}), \infty),
\end{equation}
of the set $[1, W(r, k^{\prime})]$ in $\mathbb{R}^{+}$ has to contain our other van der Waerden number $W(r, k)$, where
\begin{equation}
k \in [\sqrt{n + 1}, \infty),
\end{equation}
and where 
\begin{equation}
[\sqrt{n + 1}, \infty) = \mathbb{R}^{+} - [1, \sqrt{n + 1}).
\end{equation}
For if we assume instead that \(k \in [1, \sqrt{n + 1})\) then both \(W(r, k^{\prime}) \leq W(r, k)\) and \(k^{\prime} \leq k\) are possible and we will get a contradiction, since we were given the strict inequality \(k^{\prime} < k\) and also that \(W(r, k^{\prime}) < r^{n} \leq W(r, k)\). Hence since our assumption leads to a contradiction it follows that \(k \in [\sqrt{n + 1}, \infty)\).\\
\indent Second since $W(r, k^{\prime})$ is a van der Waerden number smaller than $W(r, k)$ for any integer \(k^{\prime} < \sqrt{n + 1}\) but where also \(k^{\prime} < k\), the van der Waerden number $W(r, k)$ that is larger than $W(r, k^{\prime})$ must be in some larger interval $[1, W(r, k)]$ on $\mathbb{R}$, where
\begin{equation}
[1, W(r, k^{\prime})] \subset [1, W(r, k)].
\end{equation}    
Finally since \(k \in [\sqrt{n + 1}, \infty)\) must hold since \(k \in [1, \sqrt{n + 1})\) cannot hold due to how $k^{\prime}$, $k$, $W(r, k^{\prime}$ and $W(r, k)$ are defined in the Theorem, the inequality 
\begin{equation}
W(r, k) < r^{k^{2}}
\end{equation}
really must be true for any \(k \in [\sqrt{n + 1}, \infty)\) since \(k \geq \sqrt{n + 1} \Longrightarrow k^{2} \geq n + 1\) also implies Eqtn. (14), since
\begin{equation}
r^{n} \leq b_{n}r^{n} + b_{n - 1}r^{n - 1} + \cdots + b_{0} < r^{n + 1} \leq r^{k^{2}},
\end{equation}
where by substitution by $W(r, k)$ we get \(r^{n} \leq W(r, k) < r^{n + 1} \leq r^{k^{2}}\), since \(W(r, k) = b_{n}r^{n} + b_{n - 1}r^{n - 1} + \cdots + b_{0}\).
\end{proof}
Another result we get from Theorem 3.1 is that \(W(r, k) < r^{n + 1} \leq r^{k^{2}}\) also is true for any $k$, $n$ such that
$$
\sqrt{n} \leq \sqrt{\frac{W(r, k)}{\log r}} < \sqrt{n + 1} \leq k. 
$$  
Or we can take logarithms to base ten in Eqtn. (15) to derive (See Corollary 6.1, Section 6)
$$
\frac{\log_{10} W(r, k)}{\log_{10} r} < n + 1 \leq k^{2} - 1 \Longrightarrow \frac{W(r, k)^{\frac{1}{\log_{10} r}}}{10} < 10^{n} \leq 10^{k^{2} - 1}.
$$
This gets us another upper bound on $W(r, k)$, expressed in powers of ten, namely
$$
W(r, k) < (10^{n + 1})^{\log_{10} r} \leq (10^{k^{2}})^{\log_{10} r}.
$$
\section{Two different integer Power Expansions for $W(r, k)$ that sum to the same integer $W(r,k)$}
Let \(N = W(r, k) > k\). Then we also have 
\begin{eqnarray}
r^{n}&\leq&W(r,k)\\
     &=   &b_{n}r^{n} + b_{n - 1}r^{n - 1} + \cdots + b_{0} \in [r^{n}, r^{n + 1})\nonumber\\
     &=   &c_{m}k^{m} + c_{m - 1}k^{m - 1} + \cdots + c_{0}\nonumber\\
     &<   &r^{n + 1},
\end{eqnarray}
since both 
\begin{eqnarray}
W(r, k)&=   &b_{n}r^{n} + b_{n - 1}r^{n - 1} + \cdots + b_{0} \in [r^{n}, r^{n + 1}),\nonumber\\
W(r, k)&=   &c_{m}k^{m} + c_{m - 1}k^{m - 1} + \cdots + c_{0} \in [k^{m}, k^{m + 1}),
\end{eqnarray}
on the right hand sides, comprise two different finite series of power expansions, but of the same integer $W(r, k)$ for some nonnegative integers 
\begin{equation}
c_{m}, c_{m - 1}, \cdots c_{0} \in [0, k - 1], \: c_{m} \in [1, k - 1].
\end{equation}
The two different expansions in Eqtn. (18) sum precisely to the same integer $W(r, k)$. So we have
\begin{eqnarray}
W(r, k)&\in            &[k^{m}, k^{m + 1}), W(r, k) \in [r^{n}, r^{n + 1})\nonumber\\
       &\Longrightarrow&W(r, k) \in [k^{m}, k^{m + 1}) \cap [r^{n}, r^{n + 1}) \not = \emptyset.\nonumber
\end{eqnarray}
Briefly we show that, if the set $[k^{m}, k^{m + 1})$ and the intersection \([k^{m}, k^{m + 1}) \cap [r^{n}, r^{n + 1})\) both contain some identical integer power of $r$, then they can contain no integer power of $r$ other than $r^{n}$. For if the set $[k^{m}, k^{m + 1})$ does contain some integer power of $r$ that also is in the set $[r^{n}, r^{n + 1})$, then that integer power of $r$ must be $r^{n}$ because there is no other integer power of $r$ in the set $[r^{n}, r^{n + 1})$ since \(r^{n + 1} \not \in [r^{n}, r^{n + 1})\). 
\section{Van der Waerden Numbers $W(r^{\prime}, k)$, $W(r, k)$, where \(W(r^{\prime}, k) < W(r, k)\), \(r^{\prime} < r\)}   
Here we extend our results through the comparison of two different van der Waerden numbers $W(r^{\prime}, k)$, $W(r, k)$, with two different colorings $r^{\prime}$ and \(r > r^{\prime}\).
\begin{theorem}
Let 
\begin{equation}
W(r, k) = b_{n}r^{n} + b_{n - 1}r^{n - 1} + \cdots + b_{0},
\end{equation}
as before in Theorem 3.1, where 
\begin{equation}
r^{n} \leq W(r, k) < r^{n + 1}. 
\end{equation}
However this time let $W(r^{\prime}, k)$ be another van der Waerden number such that for integer $r^{\prime}$, \(W(r^{\prime}, k) < W(r, k)\) is true \(\forall r^{\prime} < r\), so that each of the intervals $[1, W(r^{\prime}, k)]$ and $[1, W(r, k)]$ contains an AP with the same integer $k$ number of terms. Suppose \(n' \leq n\), where $n'$ is the least positive integer exponent such that $r'^{n}$ divides $W(r^{\prime}, k)$ but that $r'^{n + 1}$ does not divide $W(r^{\prime}, k)$, and where $W(r^{\prime}, k)$ has an expansion into powers of $r^{\prime}$ just as $W(r, k)$ has an expansion into powers of $r$. Then if $k$ is the arbitrary number of terms in an AP among integers somewhere in the interval $[1, W(r, k)]$, the inequality 
\begin{equation}
W(r, k) < r^{k^{2}},
\end{equation}
is true only if for any such integer $k$, \(k \in [\sqrt{n + 1}, \infty)\).
\end{theorem}
\begin{proof}
We have that \(r^{\prime} < r, n' \leq n\) \(\Longrightarrow r^{\prime n'} < r^{n} \leq W(r, k) < r^{n + 1}\). Therefore
\begin{equation}
r^{\prime n'} \leq W(r^{\prime}, k) < W(r, k) < r^{n + 1} \: \forall r > r^{\prime},
\end{equation}
from which it follows by necessity that the inequality \(W(r, k) < r^{k^{2}}\) is true only if for all integer \(k \in [\sqrt{n + 1}, \infty)\), \(W(r, k) < r^{n + 1} \leq r^{k^{2}}\).
\end{proof}
\subsubsection{The Application of these Results to the Conjecture \(W(2, k) < 2^{k^{2}}\)}
By the Application of Theorem 3.1 and Theorem 5.1 (See also Corollary 6.1 and Proof and the comments in the subsequent paragraph), this conjecture is true if for any positive integer $k$, such that the interval $[1, W(2, k)]$ has an AP of $k$ terms,
$$
k \in [\sqrt{n + 1}, \infty),
$$
because then when this condition is met,
$$
2^{n} \leq W(2, k) < 2^{n + 1} \leq 2^{k^{2}}
$$
is true, where the integer exponent $n$ which always exists for each $W(r, k)$ (a fact which we hope the reader can see) has been defined already in Section 2.\\
\indent Actually we have even a much stronger result which follows as a Corollary.
\begin{corollary}
Let \(W(r, k) \in [r^{n}, r^{n + 1})\) as in the Proof to Theorem 3.1, and let \(k > k^{\prime}\) be true for any integer \(k^{\prime} \in [1, \sqrt{n}]\). Then \(W(r, k) < r^{k^{2}}\) actually is true for all integer \(k \geq \sqrt{n + 1}\). 
\end{corollary}
\begin{proof}
\begin{eqnarray}
k^{\prime} \leq \sqrt{n}&\Longrightarrow&r^{k^{\prime 2}} \leq r^{n}\nonumber\\
                        &\Longrightarrow&r^{k^{\prime 2}} \leq r^{n} \leq W(r, k) < r^{n + 1}.\nonumber
\end{eqnarray}
Now with $k^{\prime}$ as given, \(k \in [1, \sqrt{n}]\) is impossible since we are given that \(k > k^{\prime}\) for all \(k^{\prime} \in [1, \sqrt{n}]\). Similarly \(k^{2} = n\) also is impossible, since \(k > k^{\prime} \: \forall \: k^{\prime} \in [1, \sqrt{n}]\) \(\Longrightarrow k^{2} > k^{\prime 2} \: \forall \: k^{\prime 2} \in [1, n]\) \(\Longrightarrow k^{2} \not \in [1, n]\). Furthermore the perfect square $k^{2}$ cannot lie within the open interval $(n, n + 1)$ in the set $\mathbb{R}^{+}$, because this open interval contains no integers. We must conclude then that if \(k^{\prime 2} < k^{2}\) is true for all \(k^{\prime 2} \in [1, n]\) then \(k^{2} > k^{\prime 2}\) must be true for all \(k^{2} \in \mathbb{R}^{+} - [1, n]\), which implies 
$$
k \not \in [1, \sqrt{n}] \Longrightarrow k \in (\sqrt{n}, \infty) \Longrightarrow k^{2} \in (n, \infty).
$$   
Since \(\mathbb{R}^{+} - [1, n] = (n, \infty)\), and since $k^{2}$ is integer we really must have \(k^{2} \in [n + 1, \infty) \subset (n, \infty)\). But from this we derive at once that
\begin{eqnarray}
k^{\prime 2}&\leq           &n\leq \log_{r} W(r, k) < n + 1 \leq k^{2}\nonumber\\
            &\Longrightarrow&r^{k^{\prime 2}} \leq r^{n} \leq W(r, k) < r^{n + 1} \leq r^{k^{2}}\nonumber
\end{eqnarray}
is true for all \(k \in [\sqrt{n + 1}, \infty)\).
\end{proof}
\subsection{The Application of these Results to a lower Bound on $W(r, k)$ found by P. Erd\H{o}s and R. Rado}
\emph{All} the values of $k$ for the known van der Waerden numbers $W(2, 3)$, $W(2, 4)$, $W(2, 5)$, $W(2, 6)$, $W(3, 3)$, $W(3, 4)$ and $W(4, 3)$, are such that~\cite{Betts} 
$$
k \in [\sqrt{n + 1}, n + 1).
$$ 
So let us see what happens to values of $n$ for values of $k$ that are restricted to this set $[\sqrt{n + 1}, n + 1)$ and whether the suitable van der Waerden number $W(r, k)$ is known at present or not. By ``suitable" one means all those van der Waerden numbers $W(r, k)$ for which \(k \in [\sqrt{n + 1}, n + 1)\) is true, whether $W(r, k)$ is known currently or not. \\
\indent P. Erd\H{o}s and R. Rado~\cite{Rado} established that 
$$
W(r, k) > (2(k - 1)r^{k - 1})^{\frac{1}{2}}.
$$
So here we prove how establishing a condition on the exponent $n$ leads to \(W(r, k) \geq r^{n} > (2(k - 1)r^{k - 1})^{\frac{1}{2}}\).
\begin{theorem}
Let \(\frac{\log 2}{2\log r} = o(1), n \gg 1\), and both
\begin{equation}
\frac{\log(k - 1)}{2\log r} = o\left(\frac{k - 1}{2}\right),
\end{equation}
and \(W(r, k) > (2(k - 1)r^{k - 1})^{\frac{1}{2}}\) hold. Let
\begin{equation}
n > \frac{1}{2\log r}(\log 2 + \log (k - 1)) + \frac{k - 1}{2},
\end{equation}
hold for any \(k \in [\sqrt{n + 1}, n + 1)\). Then the following inequality
\begin{equation}
W(r, k) \geq r^{n} > (2(k - 1)r^{k - 1})^{\frac{1}{2}},
\end{equation}
is true.
\end{theorem}
\begin{proof}
Based upon what is given we have \(\forall k \in [\sqrt{n + 1}, n + 1)\) and given 
$$
n > \frac{1}{2\log r}(\log 2 + \log k - 1) + \frac{k - 1}{2},
$$
and since \(\frac{\log 2}{2\log r} = o(1)\) and \(\frac{2\log k - 1}{2(k - 1)\log r} = o(1)\) \(\Longrightarrow \frac{\log 2}{2\log r} + \frac{\log k - 1}{2\log r} < \frac{k - 1}{2}\), we obtain
\begin{eqnarray}
k \in [\sqrt{n + 1}, n + 1)&\Longrightarrow&n + 1 > k\\
                           &\Longrightarrow&n > k > \frac{k - 1}{2} + \frac{k - 1}{2} = k - 1\nonumber\\
                           &\Longrightarrow&n > \frac{\log 2}{2\log r} + \frac{\log k - 1}{2\log r} + \frac{k - 1}{2}\nonumber\\
                           &\Longrightarrow&n\log r > \frac{1}{2}\left(\log 2 + \log (k - 1) + (k - 1)\log r\right)\nonumber\\
                           &\Longrightarrow&r^{n} > (2(k - 1)r^{k - 1})^{\frac{1}{2}}\nonumber\\
                           &\Longrightarrow&W(r, k) \geq r^{n} > (2(k - 1)r^{k - 1})^{\frac{1}{2}},
\end{eqnarray}
where we have used the fact that \(W(r, k) \geq r^{n}\) is true since \(W(r, k) \in [r^{n}, r^{n + 1})\) follows from Theorem 5.1.
\end{proof}
\section{Van der Waerden Numbers $W(r, n)$, $W(2, 7)$}
Here we look at van der Waerden numbers $W(r, n)$, that is, when \(k = n\). For any van der Waerden number of this type,
$$
r^{k} \leq W(r, k) < r^{k + 1} < r^{k^{2}}, \: \forall k \geq \sqrt{k + 1}.
$$
All the van der Waerden numbers $W(2, 3)$, $W(3, 3)$ and $W(4, 3)$ are of this type~\cite{Betts} (See Tables in the Preprint), that is, \(k = n\). We demonstrate an application of our result.\\
\indent Elsewhere in a Preprint we have shown~\cite{Betts} (Corollary 2.3. See also the Tables), that if \(k \geq n\) the inequality \(W(r, k) < r^{k^{2}}\) follows immediately, since \(k \geq n\) \(\Longrightarrow W(r, k) < r^{n + 1} < r^{n^{2}} \leq r^{k^{2}}\). The three van der Waerden numbers
\begin{equation}
W(2, 3) = 9, W(3, 3) = 27, W(4, 3) = 76,
\end{equation}
in fact do meet this criterion, since in all these three cases \(k = n = 3\). Any future van der Waerden numbers of the form $W(r, k)$ that are discovered and for which \(W(r, k) < r^{k^{2}}\) is true also will meet the criterion \(k \geq \sqrt{n + 1}\), for reasons that relate to the sizes of $k$ and the positive integer exponent $n$, which we already have indicated in the previous Sections.\\
\indent From the Table by Rabung and Lotts~\cite{Rabung and Lotts}, it is evident that
\begin{equation}
W(2, 3) < W(3, 3) < W(4, 3) < \cdots < W(r - 1, 3) < W(r, 3) < \cdots
\end{equation}
We can use the results from Theorem 5.1 to find possible upper bounds on the values for $W(5, 3)$ and $W(6, 3)$ (which in the Table by Rabung and Lotts have respective lower bounds of $170$ and $223$), since each van der Waerden number we have found is bounded as
\begin{equation}
r^{n} \leq W(r, k) < r^{n + 1}. 
\end{equation}
We have 
\begin{equation}
3^{3} \leq W(3, 3) < 4^{3} \leq W(4, 3) < W(5, 3) < W(6, 3).
\end{equation}
From their table \(W(5, 3) > 170\) and \(W(6, 3) > 223\). We also have \(k = n = 3\) is true for each of $W(2, 3)$, $W(3, 3)$ and $W(4, 3)$. Note that $5^{3}$ divides the two lower bounds $170$ on $W(5, 3)$ and $223$ on $W(6, 3)$ respectively, but $5^{4}$ does not divide either of these two lower bounds. Assume \(k = n = 3\) is true for both $W(5, 3)$ and $W(6, 3)$. Then if this assumption is a correct one, 
\begin{eqnarray}
5^{3}&<&170 < W(5, 3) < 5^{4} < 5^{3^{2}},\\
6^{3}&<&223 < W(6, 3) < 6^{4} < 6^{3^{2}},
\end{eqnarray}
where \(5^{4} = 625\), \(5^{9} = 1953125\), \(6^{4} = 1296\) and \(6^{9} = 10077696\).\\
\indent With this numerical result in mind here we offer the following related corollary, followed by a statement that the reader actually can prove (See Table A):
\begin{corollary}
Let \(k \gg 1\) and suppose either Theorem 3.1 or Theorem 5.1 holds. Then \(n \in [1, k^{2} - 1]\).
\end{corollary}
\begin{proof}
\begin{eqnarray}
k \geq \sqrt{n + 1}&\Longrightarrow&k^{2} \geq n + 1 \Longrightarrow n \leq k^{2} - 1\\
                   &\Longrightarrow&n \in [1, k^{2} - 1].
\end{eqnarray}
\end{proof}
Actually we have a result that the unbiased reader can prove: 
\begin{quote}
\emph{A necessary and sufficient condition for which} \(W(r, k) < r^{n + 1} \leq r^{k^{2}}\) \emph{is true for any $k$ such that the interval $[1, W(r, k)]$ has an AP of $k$ terms, is that}
\begin{equation}
k \in [\sqrt{ n + 1}, \infty) \Longleftrightarrow n \in [1, k^{2} - 1].
\end{equation}
\end{quote}
\begin{center}
\begin{tabular}{|l      |c      |c                  |c            |c                         |c            |c                |c              |c             |c              |r|}
\hline
                $r$  &  $k$  &  $\sqrt{n + 1}$  &   $n$     &      $\log_{r} W(r, k)$   &    $n + 1$   &    $r^{n}$     &    $W(r, k)$   &   $r^{n + 1}$  &  $r^{k^{2}}$ \\   
\hline
                $2$  &  $3$  &  $2$             &   $3$     &      $3.170$              &    $4$       &    $2^{3}$     &    $9$         &   $2^{4}$      &  $2^{9}$  \\
 
                $2$  &  $4$  &  $2.449$         &   $5$     &      $5.129$              &    $6$       &    $2^{5}$     &    $35$        &   $2^{6}$      &  $2^{16}$  \\
                 
                $2$  &  $5$  &  $2.828$         &   $7$     &      $7.475$              &    $8$       &    $2^{7}$     &    $178$       &   $2^{8}$      &  $2^{25}$ \\ 

                $2$  &  $6$  &  $3.316$         &   $10$    &      $10.144$             &    $11$      &    $2^{10}$    &    $1132$      &   $2^{11}$     &  $2^{36}$  \\
                 
                $3$  &  $3$  &  $2$             &   $3$     &      $3.000$              &    $4$       &    $3^{3}$     &    $27$        &   $3^{4}$      &  $3^{9}$  \\
                
                $3$  &  $4$  &  $2.449$         &   $5$     &      $5.170$              &    $6$       &    $3^{5}$     &    $293$       &   $3^{6}$      &  $3^{16}$  \\

                $4$  &  $3$  &  $2$             &   $3$     &      $3.123$              &    $4$       &    $4^{3}$     &    $76$        &   $4^{4}$      &  $4^{9}$  \\
  
\hline
\end{tabular}
\end{center}
\begin{center}
Table A.
\end{center}
\subsection{Possible values of Exponent $n$ for Van der Waerden Number $W(2, 7)$}
\indent All the known van der Waerden numbers $W(2, 3)$, $W(2, 4)$, $W(2, 5)$, $W(2, 6)$, $W(3, 3)$, $W(3, 4)$ and $W(4, 3)$, are such that \(W(r, k) < r^{n + 1} \leq r^{k^{2}}\). This is not a mere coincidence for particular cases! It is precisely because we have the above necessary and sufficient condition 
\begin{equation}
k \in [\sqrt{ n + 1}, \infty) \Longleftrightarrow n \in [1, k^{2} - 1],
\end{equation}
satisfied by these known van der Waerden numbers. The condition \(k \in [\sqrt{n + 1}, \infty)\) is known to be true for all the van der Waerden numbers $W(2, 3)$, $W(2, 4)$, $W(2, 5)$, $W(2, 6)$, $W(3,3)$, $W(3, 4)$ and $W(4, 3)$~\cite{Betts} (Table 1). For all these seven cases it just so happens that, for all these seven van der Waerden numbers, \(k \in [3, 6]\) and \(n \in [1, 35]\). Elsewhere~\cite{Betts}, we have provided a Table to show the value of $n$, $\sqrt{n + 1}$ for each of these known van der Waerden numbers. For instance van der Waerden number \(W(2, 3) = 9\) has \(n = k = 3, n + 1 = 4\), \(k^{2} = 3^{2} = 9\), \(3^{2} - 1 = 8, 3 \geq \sqrt{4} = 2\), \(n = 3 \in [1, 8]\) and van der Waerden number $W(2, 6)$ has the value $1132$ where \(n = 10, 6^{2} = 36 > 11, 6 > \sqrt{11}\), \(n = 10 \in [1, 35]\)~\cite{Betts} (See Table 1). Therefore there certainly are already seven van der Waerden numbers known that do confirm the truthfulness of the conditions we specified in Theorem 3.1, Theorem 5.1, Corollary 6.1 and the paragraph after that Corollary and Proof, meaning
\begin{equation}
W(r, k) < r^{n + 1} \leq r^{k^{2}} \Longleftrightarrow k \in [\sqrt{n + 1}, \infty) \Longleftrightarrow n \in [1, k^{2} - 1].
\end{equation}
\subsubsection{Conditions are not Assumptions}
\indent Theorem 3.1, Theorem 5.1 and Corollary 6.1 establish actual \emph{conditions} for which \(W(r, k) < r^{n + 1} \leq r^{k^{2}}\) will be true and the comments after Corollary 3.1 indicate there is a necessary and sufficient condition for which \(W(r, k) < r^{n + 1} \leq r^{k^{2}}\), \emph{not assumptions}, as someone and remarkably so, might claim. With an assumption one assumes something is true then one extrapolates or derives some desired result from that assumption, whereas a condition actually establishes when a statement will be true. If $t$ is a real number and one is told ``\(t^{2} > 0\)" then one can assume correctly that $t$ is not zero to conclude from that assumption that the square of $t$ will not be zero. But if one is told ``Let $t$ be any real number not equal to zero" then that establishes an actual condition for which \(t^{2} > 0\) is true.\\
\indent Therefore we have found \emph{conditions} for which \(W(2, 7) < 2^{49}\) would be true, not assumptions, namely
\begin{equation}
2^{n} \leq W(2, 7) < 2^{n + 1} \leq 2^{49} \Longleftrightarrow 7 \in [\sqrt{n + 1}, \infty) \Longleftrightarrow n \in [1, 48].
\end{equation}
From Theorem 3.1 there is some positive integer exponent $n$, integer \(b_{n} = 1\) and integers \(b_{n - 1}, \ldots, b_{0} \in \{0, 1\}\), for which
\begin{equation}
2^{n} \leq W(2, 7) = 2^{n} + b_{n - 1}2^{n - 1} + \cdots + b_{0} < 2^{n + 1}.
\end{equation}
If the necessary and sufficient condition \(7 \in [\sqrt{ n + 1}, \infty) \Longleftrightarrow n \in [1, 48]\) is violated by van der Waerden number $W(2, 7)$, that is, if \(7 < \sqrt{n + 1}\) actually holds for van der Waerden number $W(2, 7)$, then the statement ``\(2^{n} \leq W(2, 7) < 2^{n + 1} \leq 2^{49}\)" is false, because then the implication is \(7 < \sqrt{n + 1} \Longrightarrow 49 < n + 1\) \(\Longrightarrow 2^{49} < 2^{n + 1} \Longrightarrow\) \(2^{49} \leq W(2, 7) < 2^{n + 1}\) where \(49 \in [1, n] \subset \mathbb{R}\), since for the open interval $(n, n + 1)$ on $\mathbb{R}$, 
\begin{equation}
49 \in (n, n + 1) \subset \mathbb{R}, 
\end{equation}
is impossible, since on the real line, there are no integers, whether the integer is $49$ or not, lying anywhere between $n$ and $n + 1$. \\
\indent At present the value for $W(2, 7)$ still is unknown. Yet we now can use Theorem 3.1, Theorem 5.1 and Corollary 6.1, to use the conditions we have found to determine if \(W(2, 7) < 2^{49}\). For instance with our approach we find that \(\log_{2}W(2, 7) \in [n, n + 1)\), where if for any \(7 \geq \sqrt{n + 1}\) such that \(W(2, 7) < 2^{n + 1} \leq 2^{49}\), the possible values of $n$ for $W(2, 7)$ are \(n \in [1, 48]\) where \(48 = 7^{2} - 1\) (See Corollary 6.1 and the remarks and unproven Corollary immediately afterward). However we can find a finer bound for $n$ than $[1, 48]$. By Rabung and Lotts~\cite{Rabung and Lotts}, we see that 
\begin{equation}
2^{11} < 3703 < W(2, 7). 
\end{equation}
Let \(2^{y} = 3703\) for some positive real value exponent \(y > 11\). We see that if Theorem 5.1 and Corollary 6.1 do apply to $W(2, 7)$, then since \(2^{11} < 2^{y} = 3703 < W(2, 7)\), it would follow that \(11 < y \Longrightarrow\) \(y, n, \log_{2}W(2, 7) \in (11, 48]\). So a necessary and sufficient condition for which  
\begin{equation}
2^{11} < 2^{y} < W(2, 7) < 2^{n + 1} \leq 2^{49},
\end{equation}
will hold for van der Waerden number $W(2, 7)$ when \(k = 7\), is for, when \(k = 7\), we have \(7 \in [\sqrt{n + 1}, \infty) \Leftrightarrow n \in [1, 48]\) to hold (See Corollary 6.1 and the paragraph after its Proof). So the integer exponent $n$ for which $2^{n}$ will divide $W(2, 7)$ while $2^{n + 1}$ does not divide $W(2, 7)$ and for which
\begin{equation}
W(2, 7) = b_{n}2^{n} + b_{n - 1}2^{n - 1} + \cdots + b_{0} < 2^{n + 1},
\end{equation}
where \(b_{n} = 1, b_{n - 1}, \ldots, b_{0} \in \{0, 1\}\), will be
\begin{equation}
n \in [11, 48],
\end{equation}
since \(y = \log_{2}3703 = 11.85447\ldots\) and using the fact that \(\log_{2}2 = 1\). This compares favorably with a result we obtained previously in a Preprint~\cite{Betts} (See Preprint, Corollary 2.1, Section 2), namely the result 
$$
n > \frac{\log W(r, k)}{\log r} - 1,
$$
since from the previous result~\cite{Betts}, we get, when here we take logarithms to the base $2$ and with the use of the lower bound \(3703 < W(2, 7), 2^{11.85447\ldots} = 3703\) by Rabung and Lotts~\cite{Rabung and Lotts},
$$
n > \frac{\log_{2} W(2, 7)}{\log_{2} 2} - 1 > \frac{\log_{2}3703}{\log_{2}2} - 1 = 10.85447\ldots \Longrightarrow n \geq 11,
$$
which compares favorably for a lower bound on $n$ with the result in Eqtn. (46). \\
\indent Thus since \(W(2, 7) > 3703\) and since~\cite{Rabung and Lotts} \(3703 = 2^{11.85447\cdots}\), for the van der Waerden number $W(2, 7)$ to be smaller than $2^{49}$, it must lie somewhere within the integer set $[2^{n}, 2^{n + 1}]$ for some exponent value \(n \in [\lfloor 11.85447\rfloor, 48]\) and for some real positive exponent \(\delta(2, 7) \in (11.85447\ldots, n + 1)\), such that \(2^{n} \leq W(2, 7) = 2^{\delta(2, 7)}\), where \(7 \geq \sqrt{n + 1}\) (See Theorem 3.1). \\
\indent With the application of Theorem 3.1, Theorem 5.1 and Corollary 6.1~\cite{Betts} (See also Table 1, previous preprint) and with the right algorithm in hand, \emph{at least one would know within what integer interval one ought to look!}\\
\subsection{In what Interval $[2^{n}, 2^{n + 1}]$ does $W(2, 7)$ Lie?}
Based upon our approach in this paper and the discussion about $W(2, 7)$ in this Section, one could test one at a time the integer intervals 
\begin{equation}
[1, 2^{11}], \: \: [1, 2^{12}], \: \: [1, 2^{13}], \: \: [1, 2^{14}], \: \: [1, 2^{15}], \ldots, [1, 2^{48}],
\end{equation}
for each or for any \(n \in [11, 48]\), to determine which of these intervals on $\mathbb{R}$ contains $W(2, 7)$, since by Theorem 3.1, one of these intervals must contain $W(2, 7)$, because when one finds the right integer exponent $n$ in $[11, 48]$, the van der Waerden number $W(2, 7)$ will be no smaller than $2^{n}$ and it will be no larger than $2^{n + 1}$ (See Theorem 3.1). If one dares to guess, one can assume that both $\log W(2, 7)$ and the exponent $n$ for which \(W(2, 7) \in [2^{n}, 2^{n + 1})\) is true lie in the interval $[11, 15]$. One can subject any of the intervals in Eqtn. (47) to $2$--colorings to create the binary strings (using $0$, $1$, instead of red, blue) others such as M. Kouril and J. Paul~\cite{Kouril}, have used in Boolean satisfiability tests, which evaluate two or three blocks (i.e., three for $3$--SAT, which is in NP) of binary strings separated by parentheses and by the ``$\land$" operator and arranged in conjunctive normal form, to determine whether this evaluates to true or false for the presence of the APs of length seven. This would necessitate the use either of some good Beowulf clusters~\cite{Kouril}, or high performance computing or a computer grid, technology to which at the moment the Author has no access. Yet back in the early 2000s the Author was a volunteer in the Gimps Project to find the largest Mersenne prime among Mersenne numbers with millions of digits (www.mersenne.org). The Project used a Lucas-Lehmer algorithm along with a Fast Fourier Transform (FFT) algorithm to find the unknown Mersenne prime. Since the Mersenne numbers contained millions of digits and even hundreds of millions of digits, it involved considerable computational bit complexity. But with the advances in microprocessor computer power over the years along with the paradigm known as Moore's Law, it seems dubious that to find the right exponent $n$ such that \(2^{n} \leq W(2, 7) < 2^{n + 1}\), from somewhere in the interval $[11, 48]$ would be more demanding in terms of computation time.\\
\subsection{Van der Waerden Number $W(2, 10)$}
B. van der Waerden derived a result to show that $W(2, 10)$ had an upper bound greater than $10^{9}$~\cite{Graham and Spencer},~\cite{van der Waerden},~\cite{Graham and Rothschild}. Rabung and Lotts~\cite{Rabung and Lotts}, through the use of ``cyclic zippers" have derived a lower bound \(103474 < W(2, 10)\). Let $n$ be as defined in the previous Sections, such that
$$
W(2, 10) < 2^{n + 1}.
$$  
We get \(n + 1 > \delta(2, 10) > 16.66004\cdots \Longrightarrow n > 15.66004\cdots\), where \(16.66004\cdots = \log_{2}103474\). Suppose Theorem 5.1 and Corollary 6.1 hold for $W(2, 10)$. Then \(10 \geq \sqrt{n + 1} \Leftrightarrow n \in [1, 99] \Longrightarrow\) \(W(2, 10) < 2^{n + 1} \leq 2^{100}\) holds for $W(2, 10)$, and so for $W(2, 10)$ we get the bound 
$$
2^{16.66004\cdots} < W(2, 10) < 2^{n + 1} \leq 1267650600228229401496703205376,
$$
where \(2^{100} = 1267650600228229401496703205376\). That is,
$$
W(2, 10) \in (2^{16.66004\cdots}, 2^{100}).
$$
\indent Corollary 6.1 enables us to compare the sizes of $n$ and $r$~\cite{Betts} (See Tables), as we show with the following Corollary.
\begin{corollary}
Suppose Corollary 6.1 and Eqtn. (37) hold. Let 
\begin{equation}
W(r, k) = b_{n}r^{n} + b_{n - 1}r^{n - 1} + \cdots + b_{0} < r^{n + 1} \leq r^{k^{2}},
\end{equation}
for all \(k \geq \sqrt{n + 1}\). Suppose \(k \geq r\). Then \(n \geq r\) is possible. On the other hand suppose that \(k < r < k^{2}, k = n\). Then \(n < r < n^{2} = k^{2}\).
\end{corollary}
\begin{proof}
We see from Corollary 6.1, that \(n \in [1, k^{2} - 1]\) while \(r \in [1, k] \subset [1, k^{2} - 1]\). Since $[1, k^{2} - 1]$ is a much larger set than is $[1, k]$ and since we are given by Corollary 6.1 that $n$ is contained in $[1, k^{2} - 1]$ we conclude that \(n \geq r\) is possible.\\
\indent Now instead note that if \(k < r < k^{2}, k = n\) is true then actually \(W(r, k) = W(r, n)\), \(W(r, n) < r^{n + 1} \leq r^{n^{2}} = r^{k^{2}}\) and in addition~\cite{Betts} (See entry for $W(4, 3)$ in Table 1), \(k = n \Longrightarrow k^{2} = n^{2}\) \(\Longrightarrow k < r < k^{2}\) \(\Longrightarrow n < r < n^{2}\).
\end{proof}
\subsection{The Approximation of $W(r, k)$ by $r^{n}$, when these two Integers are very Large}
One can approximate $W(r, k)$ by $r^{n}$ with small error when these two integers are very large, as we show now.
\begin{theorem}
When both $W(r, k)$ and $r^{n}$ are very large, any approximation \(r^{n} \approx W(r, k)\) has a relative error of $|1 - O(1)|$.
\end{theorem}
\begin{proof}
\begin{eqnarray}
r^{n}&\leq&W(r, k)\\
     &=&b_{n}r^{n} + b_{n - 1}r^{n - 1} + \cdots + b_{0}\nonumber\\
     &\Longrightarrow&\left|\frac{b_{n}r^{n} + b_{n - 1}r^{n - 1} + \cdots + b_{0} - r^{n}}{b_{n}r^{n} + b_{n - 1}r^{n - 1} + \cdots + b_{0}}\right|\nonumber\\
     &=&\left|1 - \frac{r^{n}}{b_{n}r^{n} + b_{n - 1}r^{n - 1} + \cdots + b_{0}}\right| = |1 - O(1)|,
\end{eqnarray}
since one easily can demonstrate, using the fact that \(r > \max(\{b_{n}, b_{n - 1}, \cdots, b_{0}\})\) in Eqtn. (44),
\begin{eqnarray}
\frac{r^{n}}{b_{n}r^{n} + b_{n - 1}r^{n - 1} + \cdots + b_{0}}&=&\frac{1}{b_{n}\left(1 + \frac{b_{n - 1}}{b_{n}} \cdot \frac{1}{r} + \frac{b_{n - 2}}{b_{n}} \cdot \frac{1}{r^{2}} + \cdots + \frac{b_{0}}{b_{n}} \cdot \frac{1}{r^{n}}\right)}\nonumber\\
                                                              &=&\frac{1}{b_{n}(1 + O(1))} = O(1).
\end{eqnarray}
\end{proof}
In a similar manner one can show that, for large $k^{m}$, $W(r, k)$ and with
$$
W(r, k) = c_{m}k^{m} + c_{m - 1}k^{m - 1} + \cdots + c_{0},
$$
one has \(W(r, k) \approx k^{m}\), also with a relative error of \(|1 - O(1)|\).
\begin{theorem}
Let \(W(r, k) = r^{\delta(r,k)}\) where \(\delta(r,k) = \log_{r}W(r, k) \in [n, n + 1)\). Then for large $W(r, k)$, $r^{n}$, the relative error in any approximation \(r^{n} \approx W(r, k)\) is $o(1)$.
\end{theorem}
\begin{proof}
Recalling that \(n \leq \delta(r, k) < n + 1\),
\begin{eqnarray}
\left|\frac{r^{\delta(r, k)} - r^{n}}{r^{\delta(r, k)}}\right|&=&\left|1 - r^{n - \delta(r, k)}\right|\\
                                                              &=&\left|1 - \frac{1}{r^{\delta(r,k) - n}}\right|\nonumber\\
                                                              &=&\left|1 - \frac{1}{r^{O(1)}}\right|\nonumber\\
                                                              &=&o(1).
\end{eqnarray}
\end{proof}
Let \(\delta(r, k) = \log_{k}W(r, k) \Longrightarrow W(r, k) = k^{\delta(r, k)}\) and let $W(r, k)$, $k^{m}$, be very large. One also can show that \(W(r, k) \approx k^{m}\) with a relative error of $o(1)$.
\section{What happens when $r$ is as large as a Tower of Exponents, and $k$ is fixed?}
With our results, what happens to the size of the exponent $n$, should $r$ have any large value, while $k$ is a fixed constant very much smaller than $r$? Our results still hold, as $n$ remains bounded above always by $\frac{\log W(r, k)}{\log r}$ even for any $r$ extremely large. Here we prove that in this case Theorem 3.1 or Theorem 5.1 and Corollary 6.1 still hold, since this case forces $n$ to remain bounded above.
\begin{theorem}
Let $c$ be any fixed integer greater than one, and let $X$ be any very large integer such that \(r = c^{X}\) also is very large, such as a tower of $c$'s. Further let \(k = k_{0}\) be fixed such that \(c^{X} \gg k_{0}\). Then if for any such $c$ and for any large integer $X$,
\begin{equation}
n = O\left(\frac{\log W(c^{X}, k_{0})}{X\log c}\right),
\end{equation}
then \(W(c^{X}, k_{0}) < (c^{X})^{n + 1} \leq (c^{X})^{k_{0}^{2}}\) remains true for any \(k_{0} \geq \sqrt{n + 1}\).
\end{theorem}
\begin{proof}
\begin{eqnarray}
r^{n}&\leq           &W(r, k) < r^{n + 1}\\
     &\Longrightarrow&n\log r \leq \log W(r, k) < (n + 1)\log r\nonumber\\
     &\Longrightarrow&n\leq \frac{\log W(r, k)}{\log r} \leq \log W(r, k) < n + 1.
\end{eqnarray}
Now letting \(r = c^{X}, k = k_{0}\) in Eqtn. (49) we see that 
\begin{equation}
n \leq \frac{\log W(c^{X}, k_{0})}{X\log c} \Longrightarrow n = O\left(\frac{\log W(c^{X}, k_{0})}{X\log c}\right).
\end{equation}
Thus even if $c^{X}$ is very large, such as a large tower of $c$'s, we still have from Eqtn. (49) that \(n = O\left(\frac{\log W(c^{X}, k_{0})}{X\log c}\right)\) $\Longrightarrow$ \(W(c^{X}, k_{0}) < (c^{X})^{n + 1} \leq (c^{X})^{k_{0}^{2}}\) for any \(k_{0} \geq \sqrt{n + 1}\).
\end{proof}
\subsection{What happens to the Size of the Exponent $n$ when \(r^{n} \gg k\)?}
Here we prove that, with the conditions established by Theorem 3.1, Theorem 5.1 and Corollary 6.1, the integer exponent $n$ remains bounded always.
\begin{theorem}
Let \(r^{n} \gg k\) and suppose the conditions in Theorem 3.1, Theorem 5.1 or Corollary 6.1 hold for van der Waerden number $W(r, k)$. Further, let \(\frac{k}{r^{n}} = o(1)\). Then the positive integer exponent $n$ for which the inequality
\begin{equation}
r^{n} \leq W(r, k) = b_{n}r^{n} + b_{n - 1}r^{n - 1} + \cdots + b_{0} < r^{n + 1} \leq r^{k^{2}},
\end{equation}
holds for all \(k \geq \sqrt{n + 1}\), is bounded below always by $\frac{\log k}{\log r} - 1$ and it is bounded above always by $k^{2} - 1$.  
\end{theorem}
\begin{proof}
\begin{eqnarray}
     k&\ll            &r^{n}\\
      &\Longrightarrow&\log k \ll n\log r \leq \log W(r, k) < (n + 1)\log r \leq k^{2}\log r\nonumber\\
      &\Longrightarrow&\log k \ll (n + 1)\log r \leq k^{2}\log r\\
      &\Longrightarrow&\frac{\log k}{\log r} - 1 \ll n \leq k^{2} - 1\nonumber\\
      &\Longrightarrow&n \in \left(\frac{\log k}{\log r} - 1, k^{2} - 1\right],
\end{eqnarray}
\end{proof}
\subsection{Further possible Areas of Exploration}
The possible values for $k$, $r$, $n$, $n + 1$, $\sqrt{n + 1}$, $\log r$, $\log W(r, k)$ and for
\begin{equation}
\delta(r, k) = \frac{\log W(r, k)}{\log r} \in [n, n + 1),
\end{equation}
all grow much more slowly on $\mathbb{R}$ than does $W(r, k)$~\cite{Betts} (See Tables in the Preprint). In fact all the possible values for $n$ for van der Waerden numbers $W(2, 3)$, $W(2, 4)$, $W(2, 5)$, $W(2, 6)$, $W(3, 3)$, $W(3, 4)$ and $W(4, 3)$, lie within the very same interval $[3, 10]$ and 
\begin{equation}
\{\delta(2, 3), \delta(2, 4), \delta(2, 5), \delta(2, 6), \delta(3, 3), \delta(3, 4), \delta(4, 3)\} \subset [3, 11),
\end{equation}
although van der Waerden numbers \(W(2, 3) = 9\) and \(W(2, 6) = 1132\) differ by no less than $10^{3}$ where \(W(2, 3), W(2, 6) \in [1, 2^{11})\). Therefore if for any given $k$, $r$, $\log r$, one is able one day to estimate the bounds, sizes or values for the numbers $n$, $n + 1$, $\log W(r, k)$, \(\sqrt{n + 1} \leq k\), \(\delta(r, k) \in [n, n + 1)\) in some way when some unknown $W(r, k)$ is bounded below as has been done already by J. Rabung and M. Lotts~\cite{Rabung and Lotts}, one might reduce actually the computational complexity of determining or estimating by computer the possible values for any $W(r, k)$ that has a value currently unknown. The problem would be reduced to estimating instead the possible bounds on $\log_{r} W(r, k)$ or on $\log_{e} W(r, k)$, instead of searching for bounds on $W(r, k)$ and possibly using Boolean satisfiability arguments to do so.  

\pagebreak

\end{document}